\documentclass[runningheads]{llncs}
\usepackage[T1]{fontenc}
\usepackage{graphicx}
\usepackage{cite}
\usepackage{graphicx}
\usepackage{amsmath}
\usepackage{yhmath} 

\usepackage{amsfonts} 
\usepackage{amssymb}
\usepackage{color}
\usepackage{latexsym}

\usepackage{algorithm}
\usepackage[noend]{algorithmic}
\usepackage{multirow} 
\usepackage{amsmath}

\newcommand{\old}[1]{{}}

\newcommand{\D}{{{\cal{D}}}}
\newcommand{\C}{{{\cal{C}}}}
\newcommand{\F}{{{\cal{F}}}}

\begin{document}
\title{Piercing Diametral Disks Induced by Edges of Maximum Spanning Tree
\thanks{This work was partially supported by Grant 2016116 from the United States -- Israel Binational Science Foundation.}}
%
%

\author{A.\,Karim Abu-Affash\inst{1} \and
Paz Carmi\inst{2} \and
Meytal Maman\inst{2}}
\authorrunning{A. K. Abu-Affash et al.}
%
\institute{Shamoon College of Engineering, Beer-Sheva 84100, Israel \\ \email{abuaa1@sce.ac.il} \and
Department of Computer Science, Ben-Gurion University, Beer-Sheva 84105, Israel \\ \email{carmip@cs.bgu.ac.il} \\ 
\email{meytal.maman@gmail.com}
}

\maketitle              
\begin{abstract}
Let $P$ be a set of points in the plane and let $T$ be a maximum-weight spanning tree of $P$.
For an edge $(p,q)$, let $D_{pq}$ be the diametral disk induced by $(p,q)$, i.e., the disk having the segment $\overline{pq}$ as its diameter. Let $\D_T$ be the set of the diametral disks induced by the edges of $T$.
In this paper, we show that one point is sufficient to pierce all the disks in $\D_T$, thus, the set $\D_T$ is Helly. Actually, we show that the center of the smallest enclosing circle of $P$ is contained in all the disks of $\D_T$, and thus the piercing point can be computed in linear time.

\end{abstract}


\section{Introduction}
Let $P$ be a set of points in the plane and let $G=(P,E)$ be the complete graph over $P$. 
A \emph{maximum-weight spanning tree} $T$ of $P$ is a spanning tree of $G$ with maximum edge weight, where the weight of an edge $(p,q)\in E$ is the Euclidean distance between $p$ and $q$, and denoted by $|pq|$. 
For an edge $(p,q)$, let $D_{pq}$ denote the \emph{diametral} disk induced by $(p,q)$, i.e., the disk having the segment $\overline{pq}$ as its diameter.
Let $\D_T$ be the set of the diametral disks obtained by the edges of $T$, i.e., $\D_T = \{D_{pq} \ : \ (p,q) \in E_T\}$, where $E_T$ is the set of the edges of $T$. In this paper, we prove that the disks in $\D_T$ have a non-empty intersection.

\subsection{Related works} 
Let $\F$ be a set of geometric objects in the plane. A set $S$ of points in the plane \emph{pierces} $\F$ if every object in $\F$ contains a point of $S$, in this case, we say that $S$ is a \emph{piercing set} of $\F$.
The piercing problem, i.e., finding a minimum cardinality set $S$ that pierces a set of geometric objects, has attracted researchers for the past century.

A famous result is Helly's theorem~\cite{Helly23,Helly30}, which states that for a set $\F$ of convex objects in the plane, if every three objects have a non-empty intersection, then there is one point that pierces all objects in $\F$. 
The problem of piercing pairwise intersecting objects has been also studied, particularly when the objects are disks in the plane. It has been proven by Danzer~\cite{Danzer86} and by Stacho~\cite{Stacho65,Stacho814} that a set of pairwise intersecting disks in the plane can be pierced by four points. However, these proofs are involved and it seems that they can not lead to an efficient algorithm.
Recently, Har-Peled et al. ~\cite{HarPeled21} showed that every set of pairwise intersecting disks in the plane can be pierced by five points and gave a linear time algorithm for finding these points.
Carmi et al.~\cite{Carmi18} improved this result by showing that four points are always sufficient to pierce any set of pairwise intersecting disks in the plane, and also gave a linear time algorithm for finding these points.

In 1995, Fingerhut~\cite{Eppstein} conjectured that for any maximum-weight perfect matching $M = \{(a_1,b_1), (a_2,b_2), \dots, (a_n,b_n)\}$ of $2n$ points in the plane, there exists a point $c$, such that $|ca_i| + |cb_i| \le \alpha\cdot |a_ib_i|$, for every $1 \le i \le n$, where $\alpha = \frac{2}{\sqrt{3}}$. 
That is, the set of the ellipses $E_i$ with foci at $a_i$ and $b_i$, and contains all the points $x$, such that $|a_ix| + |xb_i| \le  \frac{2}{\sqrt{3}}\cdot |a_ib_i|$, have a non-empty intersection. 
Recently, Bereg et al.~\cite{Bereg19} considered a variant of this conjecture. They proved that there exists a point that pierces all the disks whose diameters are the edges of $M$. The proof is accomplished by showing that the set of disks is Helly, i.e., for every three disks there is a point in common.

\subsection{Our contribution}
A common and natural approach to prove that all the disks in $\D_T$ have a non-empty intersection is using Helly's Theorem, i.e., to show that every three disks have a non-empty intersection.
However, we use a different approach and show that all the disks in $\D_T$ have a non-empty intersection by characterizing a specific point that pierces all the disks in $\D_T$. 
More precisely, we prove the following theorem.
\begin{theorem}\label{thm:mainThm}
Let $\C^*$ be the smallest enclosing circle of the points of $P$ and let $c^*$ be its center. 
Then, $c^*$ pierces all the disks in $\D_T$.
\end{theorem}
This approach is even stronger since it implies a linear-time algorithm for finding the piercing point, using Megiddo's linear-time algorithm~\cite{Megiddo83} for computing the smallest enclosing circle of $P$.

The result in this paper can be considered as a variant of Fingerhut's Conjecture. That is, for a maximum-weight spanning tree (instead of a maximum-weight perfect matching) and $\alpha = \sqrt{2}$ (instead of $\alpha = \frac{2}{\sqrt{3}}$) the conjecture holds. 


\section{Preliminaries}
Let $P$ be a set of points in the plane, let $T$ be a maximum-weight spanning tree of $P$, and let $\D_T$ be the set of the diametral disks induced by the edges of $T$. 
Let $\C^*$ be the smallest enclosing circle of the points of $P$, and let $r^*$ and $c^*$ be its radius and its center, respectively. 
We assume, w.l.o.g., that $r^*=1$ and $c^*$ is located at the origin $(0,0)$. 
Let $\D^*$ be the disk having $\C^*$ as its boundary. 
Let $A_1$, $A_2$, $A_3$, and $A_4$ (resp., $Q_1$, $Q_2$, $Q_3$, and $Q_4$) be the four arcs (resp., the four quarters) obtained by dividing $\C^*$ (resp., $\D^*$) by the $x$ and the $y$-axis; see Figure~\ref{fig:Dc quarter} for an illustration.

\begin{lemma}\label{lemma:quarterDivision}
Each one of the arcs $A_1$ and $A_3$ contains at least one point of $P$ or each one of the arcs $A_2$ and $A_4$  contains at least one point of $P$.
\end{lemma}
\begin{proof}
By definition, there are at least two points of $P$ on $\C^*$. If there are exactly two points $p$ and $q$ on $\C^*$, then the segment $\overline{pq}$ is a diameter of $\C^*$, and clearly, $p$ and $q$ are on non-adjacent arcs of $\C^*$; see Figure~\ref{fig:Dc quarter}(a).
Otherwise, there are at least three points of $P$ on $\C^*$; see Figure~\ref{fig:Dc quarter}(b,c). In this case, there are three points $p$, $q$, and $t$ on $\C^*$, such that the triangle $\triangle pqt$ contains $c^*$. Thus, every angle in this triangle is acute, and therefore two points from $p,q,t$ are on non-adjacent arcs of $\C^*$.
\old{
If $D(P)$ is defined by the points $p_1$ and $p_2$ which are $D(P)$ diameter, then trivially they are in non-adjacent quadrants $Q_i$ and $Q_{i+2}$. Otherwise, the circle defined by three points, $p_1,p_2$ and $p_3$, such that the triangle $\triangle p_1p_2p_3$ contains $c^*$. Thus, every angle in this triangle is acute, and therefore two points of $a,b,c$ are in non-adjacent quadrants.
}
\end{proof}
\begin{figure}[htb]
    \centering
    \includegraphics[width=0.99\textwidth]{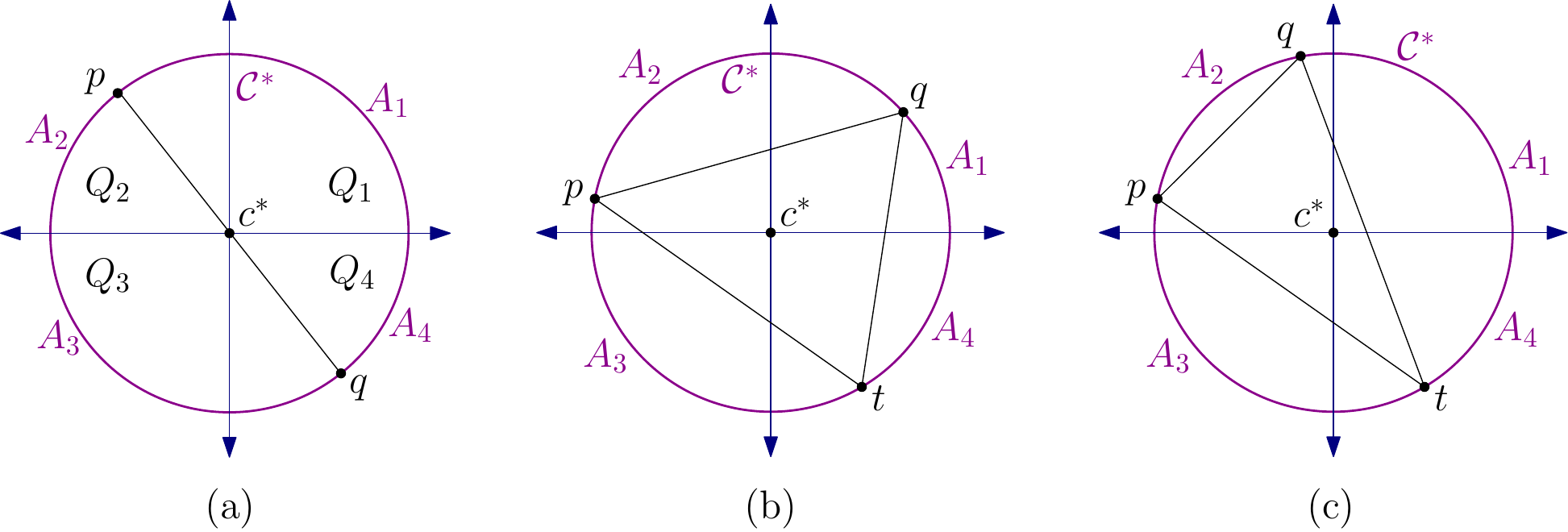}
    \caption{The smallest enclosing circle $\C^*$ of $P$. (a) Two points on $\C^*$. (b) and (c) Three points on $\C^*$.}
    \label{fig:Dc quarter}
\end{figure}


\begin{lemma}\label{lemma:fartherPoint-b-a}
Let $p$ and $q$ be two points in $Q_3$, such that $p$ is on the negative $x$-axis, the angle $\angle pc^*q<\frac{\pi}{2}$, and $|c^*p|\geq |c^*q|$; see Figure~\ref{fig:farthetLemma1}. Then,
\begin{itemize}
    \item [(i)] for every point $t$ on $A_1\cup A_2$, we have $|qt| > |pq|$,
    \item [(ii)] for every point $t$ on $A_1\cup A_4$, we have $|pt| > |pq|$, and
    \item [(iii)] for every two points $t$ and $t'$ on $A_2$ and $A_4$, respectively, we have $|tt'|>|pq|$.
\end{itemize}
\end{lemma}
\begin{proof}
\begin{itemize}
    \item [(i)] Let $a$ and $b$ be the intersection points of $\C^*$ with the negative and the positive $x$-axis, respectively; see Figure~\ref{fig:farthetLemma1}(a).
    Let $D_q$ be the disk with center $q$ and radius $|qa|$. 
    Since $|c^*q|\le |c^*p|$, we have $\angle c^*pq \le \angle c^*qp$, and thus $\angle c^*pq \le \frac{\pi}{2}$. Hence, $\angle qpa > \frac{\pi}{2}$, and thus $|qa|>|qp|$.
    Let $q'$ be the intersection point of the line passing through $a$ and $q$ with the $y$-axis, and let $D_{q'}$ be the disk with center $q'$ and radius $|q'a|$; see Figure~\ref{fig:farthetLemma1}(a).
    Since $D_{q'}$ intersects $\C^*$ at the points $a$ and $b$, the arc $A_1\cup A_2$ is outside $D_{q'}$ (this is correct for every disk centered at a point $x$ on the negative $y$-axis and has a radius $|xa|$). Thus, for every point $t$ on $A_1\cup A_2$, we have $|q't|\ge|qa|$. Since $D_q$ is contained in $D_{q'}$, this is also correct for $D_q$. 
    Therefore, for every point $t$ on $A_1\cup A_2$, we have $|qt|\ge|qa|>|qp|$.
    \begin{figure}[htb]
    \centering
    \includegraphics[width=0.87\textwidth]{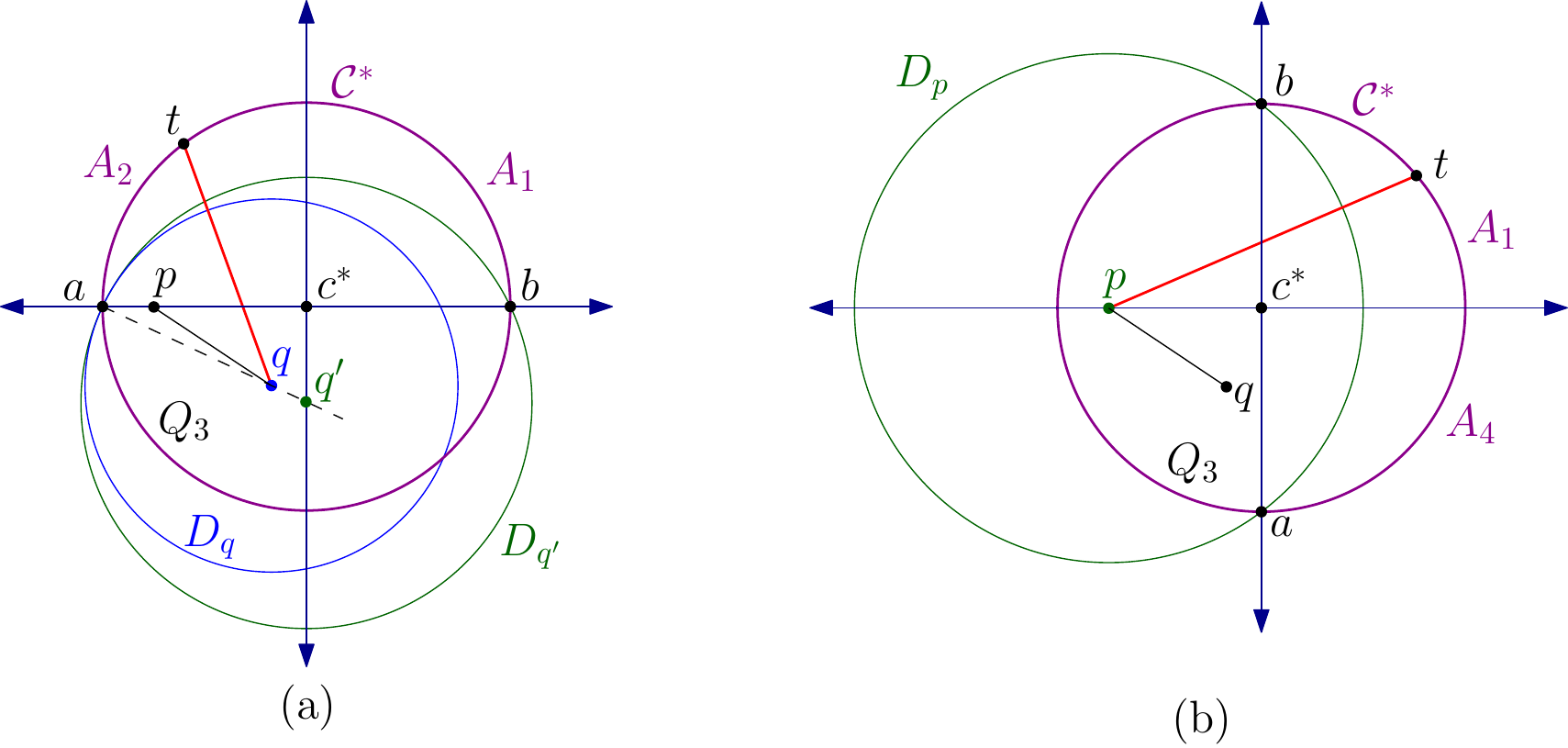}
    \caption{Illustration of the proof of Lemma~\ref{lemma:fartherPoint-b-a}. (a) Any point $t$ on $A_1\cup A_2$ satisfies $|qt| > |pq|$. (b) Any point $t$ on $A_1\cup A_4$ satisfies $|pt| > |pq|$.}
    \label{fig:farthetLemma1}
    \end{figure}
    \item[(ii)] Let $a$ and $b$ be the intersection points of $\C^*$ with the negative and the positive $y$-axis, respectively; see Figure~\ref{fig:farthetLemma1}(b). 
    Let $D_p$ be the disk centered at $p$ with radius $|pa|$. 
    Hence, $D_p$ contains $Q_3$, and thus for every point $z\in Q_3$, we have $|pz|<|pa|$, particularly $|pq|<|pa|$. 
    Since $D_{p}$ intersects $\C^*$ at the points $a$ and $b$, the arc $A_1\cup A_4$ is outside $D_p$ (this is correct for every disk centered at a point $x$ on the negative $x$-axis and has a radius $|xa|$).
    Therefore, for every point $t$ on $A_1\cup A_4$, we have $|pt|>|pa|>|pq|$.
    \item [(iii)] Since $\angle pc^*q<\frac{\pi}{2}$, we have $|pq|<\sqrt{2}$. Moreover, by the location of $t$ and $t'$, we have $|tt'|\geq\sqrt{2}$. Therefore, $|tt'|>|pq|$.
\end{itemize}
\end{proof}

Notice that Lemma~\ref{lemma:fartherPoint-b-a} holds for every two points $p$ and $q$ inside $\C^*$, such that $\angle pc^*q< \frac{\pi}{2}$. This is true since we can always rotate the points of $P$ around $c^*$ (and reflect them with respect to the $x$-axis if needed) until the farthest point from $c^*$ among $p$ and $q$ lays on the negative $x$-axis and the other point lays inside $Q_3$.  
\begin{corollary}\label{cor:farthest-a-b}
Lemma~\ref{lemma:fartherPoint-b-a} holds for every two points $p$ and $q$ inside $\C^*$, such that $\angle pc^*q< \frac{\pi}{2}$.
 
\end{corollary}


\section{Proof of Theorem~\ref{thm:mainThm}}  
Let $G=(P,E)$ be the complete graph over $P$ and let $T=(P,E_T)$ be the maximum-weight spanning tree of $P$ (i.e., of $G$).
A maximum-weight spanning tree can be computed by Kruskal's algorithm~\cite{Cormen09} (or by the algorithm provided by Monma et al.~\cite{MonmaClyde1990Cems}) which uses the fact that for any cycle $C$ in $G$, if the weight of an edge $e\in C$ is less than the weight of each other edge in $C$, then $e$ cannot be an edge in any maximum-weight spanning tree of $P$. Kruskal's algorithm works as follows. It sorts the edges in $E$ in non-increasing order of their weight, and then goes over these edges in this order and adds an edge $(p,q)$ to $E_T$ if it does not produce a cycle in $T$.
Based on this fact, we prove that for every edge $(p,q)\in E_T$, the disk $D_{pq}$ contains $c^*$. 
More precisely, we prove that for each edge $(p,q)\in E_T$ the angle $\angle pc^*q$ is at least $\frac{\pi}{2}$.
\begin{lemma}\label{lemma:cycle}
For every edge $(p,q)\in E_T$, we have $\angle pc^*q\geq \frac{\pi}{2}$.
\end{lemma}
\begin{proof}
Let $(p,q)$ be an edge in $E_T$. 
We show that if $\angle pc^*q<\frac{\pi}{2}$, then there is a cycle in $G$ in which the edge $(p,q)$ has the minimum weight among the edges of this cycle, and thus $(p,q)$ can not be in a maximum-weight spanning tree of $P$.
Assume towards a contradiction that $\angle pc^*q<\frac{\pi}{2}$, and assume, w.l.o.g., that $p$ and $q$ are in $Q_3$, $p$ is on the $x$-axis, and $|c^*p|>|c^*q|$.
We distinguish between two cases:
\begin{itemize}
    \item [(i)] If there is a point $t$ on $A_1$, then, by Lemma~\ref{lemma:fartherPoint-b-a}, we have $|tp|>|pq|$ and $|tq|>|pq|$. Thus, the edges  $(t,p)$, $(p,q)$, and $(q,t)$ form a cycle and the edge $(p,q)$ has a weight less than the weight of each other edge in this cycle; see Figure~\ref{fig:cycle proof}(a). This contradicts that $(p,q) \in E_T$. 
    \item [(ii)] Otherwise, by Lemma~\ref{lemma:quarterDivision}, there exist two points $t$ and $t'$ on $A_2$ and $A_4$, respectively. By Lemma~\ref{lemma:fartherPoint-b-a}, we have $|tq|>|pq|$, $|t'p|>|pq|$ and $|tt'|>|pq|$. Thus, the edges $(t,t')$, $(t',p)$, $(p,q)$, and $(q,t)$ form a cycle and the edge $(p,q)$ has a weight less than the weight of each other edge in this cycle; see Figure~\ref{fig:cycle proof}(b). This contradicts that $(p,q) \in E_T$.
\end{itemize}
\begin{figure}[htb]
    \centering
    \includegraphics[width=0.8\textwidth]{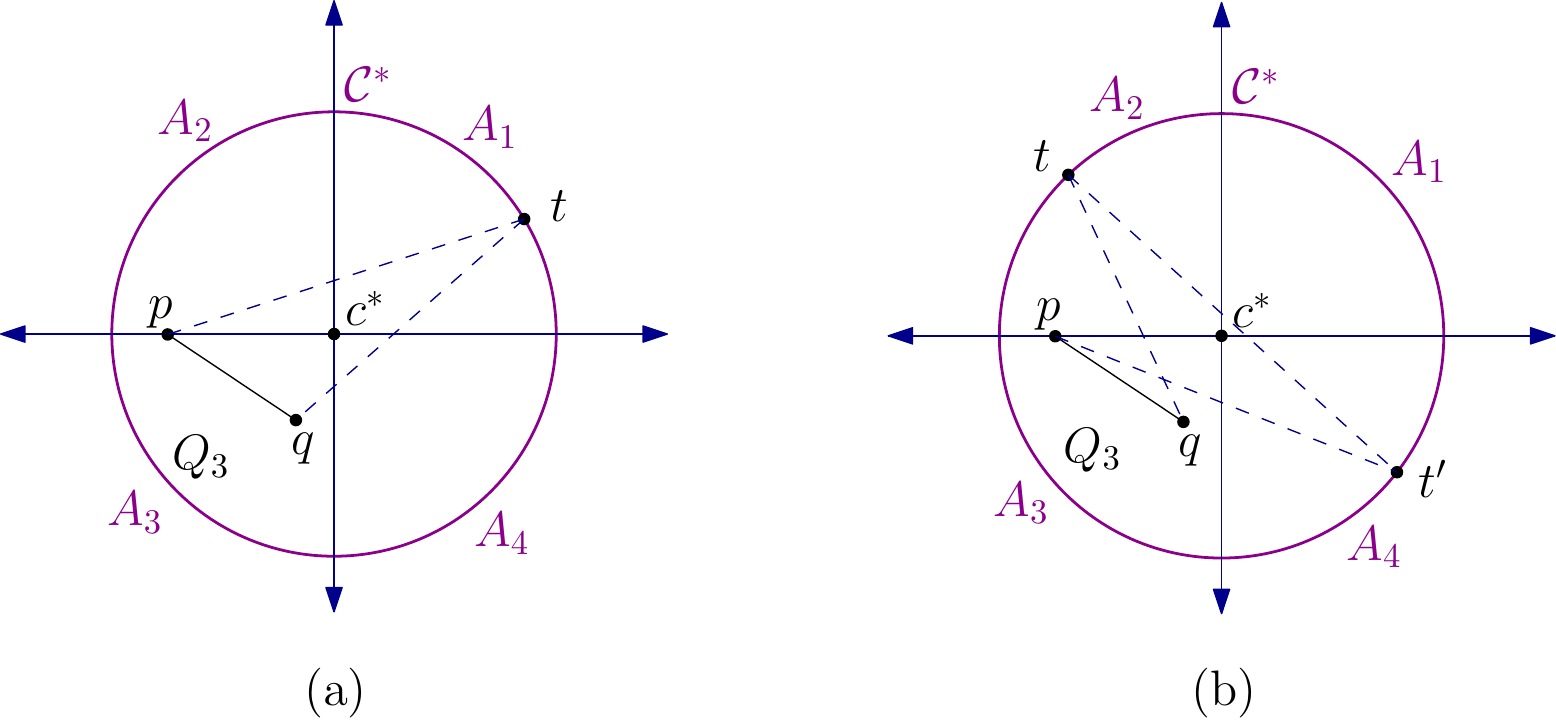}
    \caption{Illustration of the proof of Lemma~\ref{lemma:cycle}. (a) $(p,q)$ is of minimum weight in the cycle $<t,p,q>$. (b) $(p,q)$ is of minimum weight in the cycle $<t,p,q,t'>$.}
    \label{fig:cycle proof}
\end{figure}
\end{proof}

\section{Conclusion}
In this paper, we have shown that the diametral disks obtained by the edges of a maximum-weight spanning tree of a set of points $P$ have a non-empty intersection. We showed that the disks can be pierced by the center of the smallest enclosing circle of $P$, which can be computed in linear time~\cite{Megiddo83}.

Fingerhut~\cite{Eppstein} conjectured that for any maximum-weight perfect matching $M = \{(a_1,b_1), (a_2,b_2), \dots, (a_n,b_n)\}$ of $2n$ points in the plane, the set of the ellipses $E_i$ with foci at $a_i$ and $b_i$, and contains all the points $x$, such that $|a_ix| + |xb_i| \le  \alpha\cdot |a_ib_i|$, for every $1 \le i \le n$, where  $\alpha=\frac{2}{\sqrt{3}}$, have a non-empty intersection.
The smallest known value for $\alpha$ is $\alpha=\sqrt{2}$, which was provided by Bereg et al.~\cite{Bereg19}.

In this paper, we considered a variant of Fingerhut's Conjecture for maximum-weight spanning tree instead of maximum-weight perfect matching. We showed that for any maximum-weight spanning tree $T$ and $\alpha = \sqrt{2}$, there exists a point $c^*$, such that for every edge $(a,b)$ in $T$, $|c^*a| + |c^*b| \le \alpha\cdot |ab|$. In Figure~\ref{fig:MST}(a), we show an example of a maximum-weight spanning tree, such that for any $\alpha < \frac{1+\sqrt{3}}{2}$, the conjecture does not hold. This provides a lower bound on $\alpha$.
Moreover, in Figure~\ref{fig:MST}(b), we show an example of a maximum-weight spanning tree for which the center $c^*$ of the smallest enclosing circle does not satisfy the inequality for $\alpha = \frac{1+\sqrt{3}}{2}$. This means that our approach does not work for $\alpha = \frac{1+\sqrt{3}}{2}$, but does not mean that the conjecture does not hold for $\alpha = \frac{1+\sqrt{3}}{2}$.
Even though the gap between $\sqrt{2}\approx 1.414$ and $\frac{1+\sqrt{3}}{2} \approx 1.366$ is very small, it is an interesting open question to find the exact value for $\alpha$ for which the conjecture holds.
\begin{figure}[htb]
    \centering
    \includegraphics[width=0.73\textwidth]{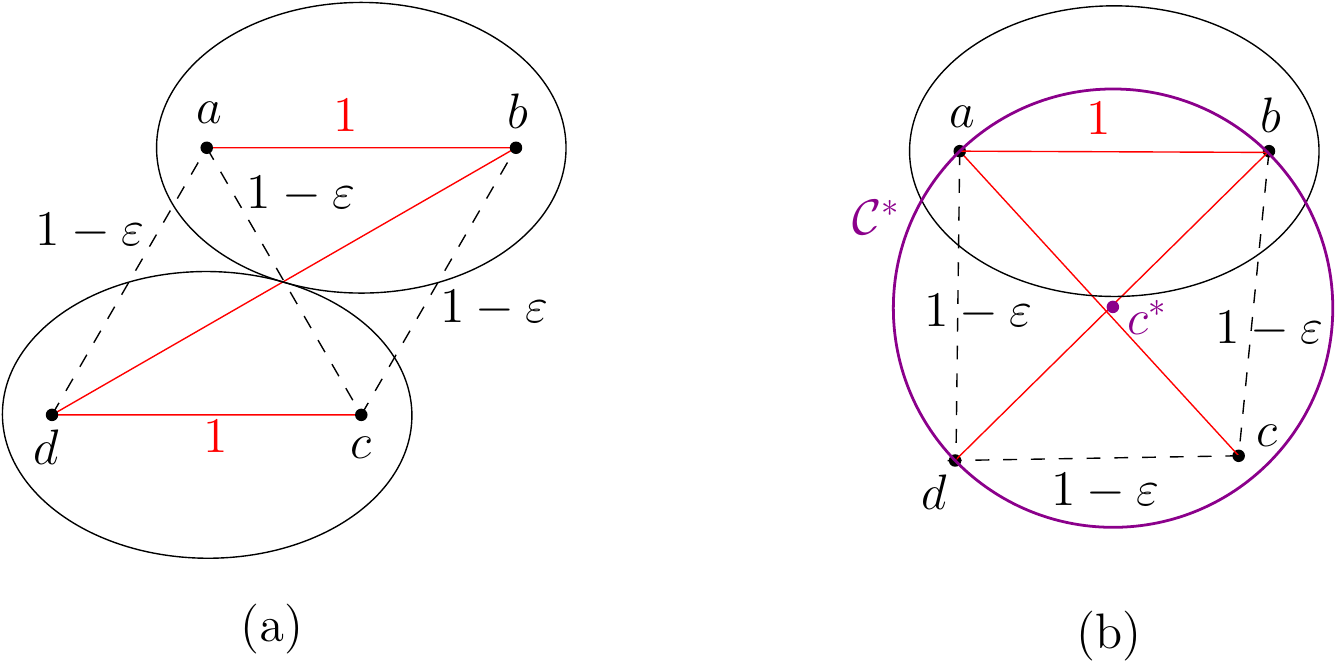}
    \caption{A maximum-weight spanning tree of the points $\{a,b,c,d\}$ (red edges) and $\alpha = \frac{1+\sqrt{3}}{2}$. (a) The ellipses defined by the edges $(a,b)$ and $(c,d)$ are tangent to each other. (b) The ellipse defined by the edge $(a,b)$ does not contain the point $c^*$.}
    \label{fig:MST}
\end{figure}

\bibliographystyle{plainurl}

\end{document}